\documentclass{article}
\usepackage{url}
\usepackage{color}
\usepackage{bm}
\usepackage{amsmath}
\usepackage{amssymb}\usepackage{amsthm}
\usepackage{amsfonts}
\usepackage{amscd}
\begin{document}
\theoremstyle{plain}
\newtheorem*{ithm}{Theorem}
\newtheorem*{idefn}{Definition}
\newtheorem{thm}{Theorem}[section]
\newtheorem{lem}[thm]{Lemma}
\newtheorem{dlem}[thm]{Lemma/Definition}
\newtheorem{prop}[thm]{Proposition}
\newtheorem{set}[thm]{Setting}
\newtheorem{cor}[thm]{Corollary}
\newtheorem*{icor}{Corollary}
\theoremstyle{definition}
\newtheorem{assum}[thm]{Assumption}
\newtheorem{notation}[thm]{Notation}
\newtheorem{defn}[thm]{Definition}
\newtheorem{clm}[thm]{Claim}
\newtheorem{ex}[thm]{Example}
\theoremstyle{remark}
\newtheorem{rem}[thm]{Remark}
\numberwithin{equation}{section}

%
%
%
\newcommand{\id}{\mathop{\mathrm{id}}\nolimits}
\newcommand{\Mat}{\mathop{\mathrm{M}}\nolimits}
\newcommand{\Tr}{\mathop{\mathrm{Tr}}\nolimits}
\newcommand{\exa}{\mathop{\mathrm{ex}}\nolimits}
\newcommand{\lmk}{\left (}
\newcommand{\rmk}{\right )}
\newcommand{\ctv}{\caC_{(\theta,\varphi)}}
\newcommand{\btv}{\caB_{(\theta,\varphi)}}
\newcommand{\bbC}{{\mathbb C}}
\newcommand{\bbR}{{\mathbb R}}
\newcommand{\bbN}{{\mathbb N}}
\newcommand{\bbT}{{\mathbb T}}
\newcommand{\bbZ}{{\mathbb Z}}
\newcommand{\caA}{{\mathcal A}}
\newcommand{\caB}{{\mathcal B}}
\newcommand{\caC}{{\mathcal C}}
\newcommand{\caD}{{\mathcal D}}
\newcommand{\caG}{{\mathcal G}}
\newcommand{\caH}{{\mathcal H}}
\newcommand{\caM}{{\mathcal M}}
\newcommand{\caO}{{\mathcal O}}
\newcommand{\caP}{{\mathcal P}}
\newcommand{\caU}{{\mathcal U}}
\newcommand{\QAut}{\mathop{\mathrm{QAut}}\nolimits}
\newcommand{\Aut}{\mathop{\mathrm{Aut}}\nolimits}
\newcommand{\Ad}{\mathop{\mathrm{Ad}}\nolimits}
\newcommand{\unit}{\mathbb I}
\newcommand{\lv}{\left \vert}
\newcommand{\rv}{\right \vert}
\newcommand{\lV}{\left \Vert}
\newcommand{\rV}{\right \Vert}
\newcommand{\braket}[2]{\left\langle#1,#2\right\rangle}
\newcommand{\Uo}{{\mathrm U}(1)}


\title{Classification of gapped ground state phases in quantum spin systems}

\author{Yoshiko Ogata \thanks{ Graduate School of Mathematical Sciences
The University of Tokyo, Komaba, Tokyo, 153-8914, Japan
Supported in part by
the Grants-in-Aid for
Scientific Research, JSPS.}}



\maketitle
\begin{abstract}
Recently, classification problems of gapped ground state phases attract a lot of attention
in quantum statistical mechanics.
We explain about our operator algebraic approach to these problems.
\end{abstract}

\theoremstyle{plain}
\section{Introduction}
In quantum mechanics, physical models are determined in terms of some self-adjoint operators
called Hamiltonians.
Recently, Hamiltonians
whose spectrum has a gap between the lowest eigenvalue (which coincides with the infimum of the spectrum) 
and the rest of the spectrum attract a lot of attention.
Physically, these models are considered to be in normal phases, where no critical phenomena
occur.
Despite that, it has turned out that the structure of these normal gapped phases
is actually mathematically interesting when we introduce some equivalence relation to them.
Roughly speaking, we say two models are equivalent if we can connect them smoothly within
those normal phases.
In spacial dimensions higher than one, it is believed (and partially proven)
there are multiple phases with respect to such classifications.
If we further introduce some symmetry to the game, we obtain interesting mathematical structures, even in one dimension.
In this talk, we explain the operator algebraic approach to those problems.

\section{Finite-dimensional quantum mechanics}
In order to motivate us for the operator algebraic framework of quantum statistical mechanics,
we first recall finite-dimensional quantum mechanics in this section.
In finite-dimensional quantum mechanics,
physical observables are represented by
elements of $\Mat_n$, the algebra of $n\times n$-matrices.
Each positive matrix $\rho$ with $\Tr\rho=1$ (called a density matrix) defines 
a physical state by
\[
\omega_\rho: \Mat_n\ni A\mapsto \Tr\lmk \rho A\rmk\in\bbC.
\]
We call this map $\omega_\rho$ a state.
Clearly, it is positive i.e., $\omega_{\rho}\lmk A^{*}A\rmk\ge 0$ and normalized $\omega_{\rho}\lmk\unit\rmk=1$.
This corresponds to the procedure of taking
{expectation values} of each {physical observables} $A\in\Mat_n$, in the 
{physical state {$\omega_\rho$}}.
Note that the set of all states forms a convex compact set.
Its extremal points are called pure states.
A state $\omega_\rho$ is pure if and only if $\rho$ is a rank one projection.

{Time evolution} (Heisenberg dynamics)
is given by {a self-adjoint matrix $H$},
called a {Hamiltonian},
via the formula 
\begin{align}\label{fst}
\Mat_n\ni A\mapsto \tau_t(A):=e^{it{H}} A e^{-it{H}},
\quad t\in\bbR.
\end{align}
Let $p$ be the spectral projection of $H$ corresponding to the 
{lowest eigenvalue}.
A state ${\omega_{\rho}}(A):=\Tr{\rho} A$ on $\Mat_n$ is said to be {a ground state} of {$H$}
if the support of $\rho$ is under $p$.
The ground state is {unique} if and only if
$ p$ is a rank one projection, i.e., if the {lowest eigenvalue} of {$H$}
is {non-degenerated}.
In this case, the {unique ground state} is of the form $\omega_{ p}(A):=\Tr{p} A$,
and it is {pure} because $p$ is rank one.

Sometimes, we consider time-dependent Hamiltonians $H(t)$.
Then the time evolution of an observable $A\in\Mat_n$
is given by a solution $\tau_t(A)$ of the differential equation
\begin{align*}
\frac{d}{dt}\tau_t(A)=i\left[ H(t), \tau_t(A)\right],\quad \tau_0(A)=A,\quad A\in \Mat_{n}.
\end{align*}
When the Hamiltonian is time-dependent $H(t)=H$, this reduces to the above Heisenberg dynamics
$e^{itH} A e^{-itH}$.

Symmetry plays an important role in physics.
Let {$G$} be a finite group and suppose that there is a group action
 {$\beta : G\to \Aut(\Mat_n)$} given by unitaries $V_g$, $g\in G$ 
 \[
 {\beta_g}(A):=\Ad\lmk{ V_g}\rmk\lmk A\rmk,\quad A\in \Mat_n,\quad g\in G.
 \]
 Here and thereafter, $\Aut(\caA)$ for a $*$-algebra $\caA$ denotes the automorphism group of $\caA$.
 If a Hamiltonian $H$ satisfies
$ {\beta_g}({H})={H}$ for all $g\in G$,
we say
$H$ is $\beta$-invariant.
If a {$\beta$-invariant} Hamiltonian $H$ has a 
{unique ground state} $\omega_{p}(A):=\Tr{p} A$,
then this {unique ground state} $\omega_{p}$  is $\beta$-invariant $\omega_{ p}\lmk \beta_g(A)\rmk=\omega_{p}(A)$, $A\in\Mat_n$, because the spectral projection $p$ is $\beta$-invariant,
i.e., $\beta_g(p)=p$.

\section{Quantum spin systems}\label{quantumspinsec}
Operator algebraic framework of quantum statistical mechanics allows us to extend the framework of finite-dimensional quantum mechanical systems to infinite dimensions.
Let $2\le d\in \bbN$ and $\nu\in\bbN$ be fixed.
Physically, $\frac{d-1}2$ denotes 
the size of on-site spin (spin quantum number)  and 
$\nu$ denotes the spacial dimension.
We denote by ${\mathfrak S}_{\bbZ^\nu}$, the set of all finite subsets of $\bbZ^{\nu}$.
For each finite subset $\Lambda\in {\mathfrak S}_{\bbZ^\nu}$, 
we associate a finite-dimensional $C^{*}$-algebra
\begin{align*}
\caA_{\Lambda}:=\bigotimes_{\Lambda}\Mat_{d}.
\end{align*}
Here, $\Mat_{d}$ is the algebra of $d\times d$ -matrices.
The $\nu$-dimensional quantum spin system $\caA_{\bbZ^\nu}$ is the 
$C^{*}$-inductive limit of this inductive net, given by the natural inclusion.
 For each infinite subset $\Gamma$, we may define $\caA_\Gamma$
in exactly the same manner.
The $C^*$-algebra $\caA_\Gamma$ can be naturally regarded as a $C^*$-subalgebra
of $\caA_{\bbZ^\nu}$.
We say an element $A$ has support in $\Gamma$ if
it belongs to $\caA_\Gamma$.
If an automorphism $\alpha$ acts trivially on $\caA_{\Gamma^c}$
for some $\Gamma\subset \bbZ^{\nu}$,
we say that $\alpha$ has support in $\Gamma$.
The set of all elements in $\caA_{\bbZ^{\nu}}$ with finite support are called local algebra and denoted by $\caA_{\mathrm loc}$.

A state $\omega$ on $\caA_\Gamma$ is defined to be a linear functional
on $\caA_\Gamma$ with $\omega(\unit)=1$
which is positive in the sense that
$\omega(A^*A)\ge 0$ for any $A\in \caA_\Gamma$.
The map $\caA_\Gamma \ni A\mapsto \omega(A)\in\bbC$
corresponds to the procedure of taking the expectation value of
a physical observable $A$ in our physical state $\omega$.
The set of all states on $\caA_\Gamma$ forms a convex weak$*$-compact
set. Its extremal points are called pure states.
By the Krein-Milman theorem, the set of states is the weak$*$-closure of the convex envelope of pure states.
See \cite{BR1} for more details.

For each state, we can associate a representation of $\caA_\Gamma$
essentially uniquely.
\begin{thm}[GNS representation]
For each state $\omega$ on $\caA_\Gamma$, there exist a
representation $\pi_{\omega}$ of $\caA_\Gamma$ on a Hilbert space $\caH_{\omega}$
and a unit vector $\Omega_{\omega}\in\caH_{\omega}$
such that
\begin{align}\label{gns}
\omega(A)=\braket{\Omega_{\omega}}{\pi_{\omega}(A)\Omega_{\omega}},\quad A\in\caA_\Gamma,\quad\text{and}
\quad
\caH_{\omega}=\overline{\pi_{\omega}(\caA_\Gamma)\Omega_{\omega}}.
\end{align}
Here, $\overline{\cdot}$ denotes the norm closure.
It is {unique up to unitary equivalence}.
\end{thm}
The triple $(\caH_\omega,\pi_\omega,\Omega_\omega)$
is called the GNS triple of $\omega$.
We frequently consider the commutant or bicommutant of
$\pi_\omega\lmk\caA_{\Gamma}\rmk$.
For a $*$-algebra $\caM$ acting on a Hilbert space $\caH$, 
we denote by $\caM'$ the set of all elements in $\caB(\caH)$ (the set of all bounded operators on $\caH$)
commuting with every element in $\caM$.
The algebra $\caM'$ is called a commutant of $\caM$, and the commutant of $\caM'$
is called bicommutant and denoted by $\caM''$.

For a pure state $\omega$, it is known that $\pi_{\omega}$ is irreducible (i.e, there is no non-trivial closed subspace
of $\caH_{\omega}$ invariant under $\pi_{\omega}(\caA_{\Gamma})$) and 
$\pi_{\omega}(\caA_{\Gamma})$
is dense in $\caB(\caH_{\omega})$ with respect to the strong operator topology.
This property can be rephrased as $\pi_\omega\lmk\caA_{\Gamma}\rmk''=\caB(\caH_{\omega})$.

Given GNS representations, we can introduce 
some equivalence relation between states.
We say two states $\omega,\varphi$ on $\caA_\Gamma$ are equivalent (denoted $\omega\simeq \varphi$)
if and only if the corresponding GNS representations are
unitarily equivalent. 
For a state $\omega$ and an automorphism $\alpha$ on $\caA_\Gamma$,
if $\omega$ and $\omega\circ\alpha$ are equivalent,
then there is a unitary $u$ on the GNS Hilbert space $\caH_\omega$
 implementing $\alpha$
 in the sense
 \begin{align}\label{implement}
 \Ad\lmk u\rmk\circ\pi_\omega=\pi_\omega\circ\alpha.
 \end{align} 
 This is because $\pi_\omega\circ\alpha$ is a GNS representation of $\omega\circ\alpha$.
In our context of quantum spin systems,
we can see that two states $\omega,\varphi$ are equivalent
  if
they can be approximated by a 
local perturbation of each other.
More precisely, $\omega$ can be approximated 
arbitrarily well in the norm topology of $\caA_{\bbZ^{2}}^{*}$ by states of the form $\varphi\lmk A^*\cdot A\rmk$,
with $A\in\caA_{\mathrm loc}$
 and vice versa.
 Physically, it means $\omega$ and $\varphi$ are macroscopically the same.
 
 There is yet another equivalence relation between states, which is called
 quasi-equivalence. Two states $\omega,\varphi$ are said to be quasi-equivalent if
 there is a $*$-isomorphism $\iota: \pi_{\omega}\lmk \caA_{\Gamma}\rmk''\to \pi_{\varphi}\lmk \caA_{\Gamma}\rmk''$
 such that $\pi_{\varphi}(A)=\iota\circ \pi_{\omega}(A)$, for all $A\in \caA_{\Gamma}$.
 Note that if two states are equivalent, they are quasi-equivalent.
 The converse is not true in general, but if the states are pure, it is true.
 
In the operator algebraic framework of quantum spin systems, physical models are specified with a map
called {interactions}.
An interaction $\Phi$ is a map 
$
\Phi: {\mathfrak S}_{\bbZ^\nu}\to \caA_{\mathrm loc}
$
satisfying
\begin{align*}
\Phi(X) = \Phi(X)^*\in \caA_{X}
\end{align*}
for all $X \in {\mathfrak S}_{\bbZ^\nu}$. 
Physically, this $\Phi(X)$ indicates an interaction term between 
spins inside of $X$.

The easiest type of interaction is an  on-site interaction, satisfying
\begin{align}\label{onsitedef}
\Phi(X)=0,\quad\text{if}\quad |X|\neq 1.
\end{align}
It means the only possibly non-zero interaction terms are of the form
$\Phi(\{\mathbf x\})$, with ${\mathbf x}\in\bbZ^\nu$.
(Here and thereafter $|X|$ indicates the number of elements in $X$.)
Note that all interaction terms commute with each other for such interactions.

Physically, we are more interested in interactions that have non-zero
interaction terms between different sites of $\bbZ^{\nu}$.
For example, let 
$\{S_j\}_{j=1,2,3}$ be generators of the irreducible representation of ${\mathfrak{su}(2)}$ on $\bbC^d$.
Then an interaction of $\caA_{\bbZ}$ given by 
\begin{align}\label{heisenberg}
\Phi(\{x,x+1\})=\sum_{j=1}^3\,S^{(x)}_jS^{(x+1)}_j,\quad x\in \bbZ
\end{align}
is called the antiferromagnetic Heisenberg chain and has been extensively studied.

Now, given an interaction,
we would like to define a dynamics on $\caA_{\bbZ^\nu}$ out of it.
In order for that, we need to assume that $\Phi$ is ``suitably local''.
The simplest condition among such is the condition of the uniformly bounded and finite range.
An interaction is 
of finite range if there exists  an $m\in {\mathbb N}$ such that
$\Phi(X)=0$ for $X$ with a diameter larger than $m$.
It is  uniformly bounded if it satisfies
$
\sup_{X\in  {\mathfrak S}_{\bbZ^\nu}}\lV
\Phi(X)
\rV<\infty
$.
We can relax this restriction extensively.
More generally, we define norms on interactions and consider interactions with finite norms.
See \cite{nsy}.

Given a suitably local interaction, we may define a $C^*$-dynamics, i.e., strongly continuous one-parameter group of
automorphisms on $\caA_{\bbZ^\nu}$.
For an interaction $\Phi$ and a finite set $\Lambda\subset {\bbZ^\nu}$, we define the local Hamiltonian on $\Lambda$
by
\begin{equation}\label{GenHamiltonian}
\lmk H_{\Phi}\rmk_{\Lambda}:=\sum_{X\subset{\Lambda}}\Phi(X).
\end{equation}
Then we consider the Heisenberg dynamics given by the local Hamiltonian $e^{it(H_{\Phi})_{\Lambda}} Ae^{-it(H_{\Phi})_{\Lambda}}$
and take the thermodynamic limit.
If our interaction $\Phi$ is suitably local, for example,  if it is a uniformly bounded finite range interaction,
 the limit
\begin{equation}\label{dyn}
\tau^t_{\Phi}(A)=\lim_{\Lambda\to\bbZ^\nu} e^{it(H_{\Phi})_{\Lambda}} Ae^{-it(H_{\Phi})_{\Lambda}},\quad
t\in \bbR,\quad A\in\caA_{\bbZ^\nu}
\end{equation}
exists and defines a dynamics $\tau_\Phi$ on $\caA_{\bbZ^\nu}$.
The reason why we consider the dynamics $\tau_\Phi$
instead of Hamiltonians is because there is no
mathematically meaningful limit of local Hamiltonians $\lmk H_{\Phi}\rmk_{\Lambda}$
as $\Lambda\to\bbZ^\nu$, while
the limit (\ref{dyn}) makes sense.
For this reason, in the operator algebraic framework of quantum statistical mechanics, we talk about dynamics instead of Hamiltonians.

For the same reason, a ground state is defined in terms of the dynamics $\tau_{\Phi}$.
\begin{def}\label{groundstatedef}
Let $\delta_\Phi$ be the generator of $\tau_\Phi$.
A state $\omega$ on $\caA_{\bbZ^\nu}$ is called an $\tau_\Phi$-ground state
if the inequality
\begin{align}\label{gsdef}
-i\omega\lmk A^*{\delta_\Phi}\lmk A\rmk\rmk\ge 0
\end{align}
holds
for any element $A$ in the domain $\caD({\delta_\Phi})$ of ${\delta_\Phi}$.
\end{def}
We occasionally say a ground state of $\Phi$ instead of a $\tau_\Phi$-ground state.
We denote by $\caG_\Phi$ the set of all ground states
of $\Phi$.
Clearly, $\caG_{\Phi}$ is a weak$*$-compact convex set, and it is known that 
its extremal points $\exa \caG_{\Phi}$
consists of pure states. (See Theorem 5.3.37 \cite{BR2}.)

Let $(\caH_\omega,\pi_\omega,\Omega_\omega)$ be the GNS triple of a $\tau_\Phi$-ground state $\omega$.
Then there exists a unique positive operator $H_{\omega,\Phi}$ on $\caH_\omega$ such that
$e^{itH_{\omega,\Phi}}\pi_\omega(A)\Omega_\omega=\pi_\omega(\tau^t_{\Phi}(A))\Omega_\omega$,
for all $A\in\caA_{\bbZ^\nu}$ and $t\in\mathbb R$.
We call this $H_{\omega,\Phi}$ the bulk Hamiltonian associated with $\omega$.
Note that $\Omega_\omega$ is an eigenvector of $H_{\omega,\Phi}$ with eigenvalue $0$. 
(See Proposition 5.3.19 \cite{BR2}.)

Let us consider the corresponding condition for a finite quantum system $\Mat_n$ with dynamics
given by a Hamiltonian $H$ (\ref{fst}).
Let $p$ be the spectral projection of $H$ corresponding to the lowest eigenvalue $E_{0}$.
Recall that 
a state $\omega$ on $\Mat_{n}$ is given by a density matrix $\rho$ 
with the formula $\omega(A)=\Tr\rho A$.
Let $s(\rho)$ be the support projection of this $\rho$.
Then one can check that $\omega$
 is a $\tau$-ground state if and only if
$s(\rho)$  satisfies $s(\rho)\le p$.
Recall that the last condition is the very definition of the ground state in finite-dimensional quantum mechanics.
In fact, note that the generator $\delta$ of $\tau$ in (\ref{fst}) is
$\delta(A)=i[H,A]$.
If $s(\rho)\le p$, then we have
\begin{align*}
-i\omega\lmk A^*{\delta}\lmk A\rmk\rmk=
\omega\lmk A^{*}(H-E_{0})A\rmk\ge 0,\quad A\in\Mat_{n},
\end{align*}
hence $\omega$ is a $\tau$-ground state.
Conversely, suppose that $\omega$ is a $\tau$-ground state. For any unit eigenvectors $\xi,\eta$ of $H$ with $H\xi=E_{0}\xi$, $H\eta=E\eta$, for $E>E_{0}$,
set $A\in\Mat_{n}$ to be a matrix satisfying $A\zeta=\braket{\eta}{\zeta}\xi$ for any $\zeta\in\bbC^{n}$.
Substituting this $A$, we get 
\begin{align*}
0\le -i\omega\lmk A^*{\delta}\lmk A\rmk\rmk
=\lmk E_{0}-E\rmk\braket{\eta}{\rho\eta}
\end{align*}
Because $E_{0}-E<0$, this means that $\braket{\eta}{\rho\eta}=0$ for any such $\eta$.
Hence we conclude that $p\rho p=\rho$, namely, $s(\rho)\le p$.
It means that our definition in operator algebraic framework can be regarded
as a natural generalization of the usual definition of a ground state to infinite systems.

Note, in general, that there can be many states satisfying the condition (\ref{gsdef}).
Namely, the ground state does not need to be unique.
If the ground state is unique, it is automatically an extremal point of $\caG_{\Phi}$.
As a result, it is pure.

The systems we are interested in, in this talk, are the ones with gapped ground states.
\begin{defn}\label{gapint}
We say $\Phi$ has gapped ground states in the bulk if
the followings hold.
\begin{description}
\item[(i)]
The bulk Hamiltonian $H_{\omega,\Phi}$ of any pure $\tau_\Phi$-ground state
$\omega$ has $0$ as its non-degenerate eigenvalue. 
\item[(ii)]
There exists a constant $\gamma>0$ such that
\begin{align}
\sigma\lmk H_{\omega,\Phi}\rmk\setminus \{0\}
\subset [\gamma,\infty),
\end{align}
for any pure $\tau_\Phi$-ground state
$\omega$.
Here $\sigma\lmk H_{\omega,\Phi}\rmk$
denotes the spectrum of $ H_{\omega,\Phi}$.
\end{description}
\end{defn}
We denote by $\caP$, the set of all uniformly bounded finite range interactions with gapped ground states in the bulk.

An interaction $\Phi$
is said to have a unique gapped ground state
if its ground state is unique and gapped in the
sense of Definition \ref{gapint}. 
See \cite{Affleck:1988vr},\cite{Fannes:1992vq},\\\cite{fnwpure},\cite{Ogata1},\cite{Ogata2},\cite{Ogata3}
for examples of such models.
If we consider the corresponding condition for a finite system $\Mat_n$ with dynamics
(\ref{fst}).
This condition corresponds to the situation
that ``the lowest eigenvalue of $H$ is non-degenerated and 
the difference between the lowest eigenvalue and the second-lowest eigenvalue is at least $\gamma$''.
%
One remarkable property of the unique gapped ground state is the exponential decay of correlation functions.
\begin{thm}[\cite{hk} \cite{ns06} \cite{ns09}]\label{expdec}
Let $\Phi$ be a uniformly bounded finite
range interaction with a {unique gapped ground state} $\omega_{\Phi}$.
Then the {correlation functions of $\omega_\Phi$
decay exponentially fast} : there exist 
constants $\mu>0$ and
$C>0$ such that for all $A\in\caA_X$,
$B\in \caA_Y$, with finite $X,Y\subset\bbZ^\nu$
\[
\lv
\omega_{\Phi}\lmk A B\rmk
-\omega_{\Phi}(A)\omega_{\Phi}(B)
\rv\le C \lV A\rV\lV B\rV |X|{e^{-\mu d(X,Y)}}
\]
hold.
Here $d(X,Y)$ denotes the distance between $X$ and $Y$.
\end{thm}
This means $\omega_\Phi$ is { ``almost like a product state"}.

\section{Paths of automorphisms generated by time-dependent interactions}
In the previous section, we considered time-independent interactions,
and derived a$C^*$-dynamics out of them.
The same procedure can be carried out for time-dependent interactions
to derive a strongly continuous paths of automorphisms.
(Recall that in finite dimensional quantum mechanics,
we also considered time-dependent Hamiltonians.)
Let $\Phi : [0,1]\ni t\to \Phi_t=(\Phi(X;t))$ be a piecewise continuous path of interactions.
Namely, for each finite $X$, the matrix-valued function
$[0,1]\ni t\to \Phi(X;t)\in\caA_{X}$ is piecewise continuous.
We then define the path of local Hamiltonians
$\lmk H_{\Phi_t}\rmk_{\Lambda}:=\sum_{X\subset{\Lambda}}\Phi(X;t)$
for each finite subset $\Lambda$ of $\bbZ^\nu$
and 
consider the solution $\alpha_{\Phi,t,\Lambda}(A)$
of the differential equation
\begin{align*}
\frac{d}{dt}\alpha_{\Phi,t,\Lambda}(A)=i\left[\lmk H_{\Phi_t}\rmk_{\Lambda}, \alpha_{\Phi,t,\Lambda}(A)\right],\quad 
\alpha_{\Phi,0,\Lambda}(A)=A.
\end{align*}
If the interactions along this path are suitably local,
analogous to the ones considered in the previous 
section, then the thermodynamic limit 
\begin{align*}
\alpha_{\Phi,t}(A)=\lim_{\Lambda\to\bbZ^\nu} \alpha_{\Phi,t,\Lambda}(A),\quad A\in\caA_{\bbZ^\nu}
\end{align*}
exists and defines a strongly continuous path of automorphisms $\alpha_{\Phi,t}$.
We denote by $\QAut\lmk \caA_{\bbZ^{\nu}}\rmk$
the set of all automorphisms $\alpha=\alpha_{\Phi,t}$ generated by some time-dependent interactions $\Phi$ in this manner.
It forms a subgroup of the automorphism group $\Aut(\caA_{\bbZ^\nu})$ on
$\caA_{\bbZ^\nu}$.

Due to the fact that $\alpha\in \QAut(\caA_{\bbZ^{\nu}})$ is given out of local interactions,
it shows some nice locality properties.
The most famous one is the Lieb-Robinson bound, which has been extensively studied and used \cite{hk} \cite{ns06} \cite{ns09}
\cite{bmns} \cite{NSY}.
It gives an estimate on $\lV
\left[
\alpha(A), B
\right]
\rV$
for 
$A\in \caA_{X}$, $B\in\caA_{Y}$, which decays as the distance between finite subsets $X$ and $Y$
goes to infinity.

The other property that is satisfied by $\alpha\in \QAut(\caA_{\bbZ^{\nu}})$ is the factorization property.
It basically says that we can split $\alpha$ into two along any cut of the system
modulo some error terms localized around the boundary.
For example,  in one-dimensional systems, if we cut the system into two parts at the origin, we have
\begin{align}\label{ofac}
\alpha=\Ad(v)\circ\lmk\alpha_{L}\otimes\alpha_{R}\rmk,
\end{align}
where $\alpha_{L}$ is an automorphism on the left infinite chain $\caA_{L}:=\caA_{(-\infty,-1]\cap\bbZ}$,
$\alpha_{R}$ an automorphism on 
the right infinite chain $\caA_{R}:=\caA_{[0,\infty)\cap\bbZ}$.
The term $\Ad(v)$ is an inner automorphism given by some
unitary $v$ in $\caA_{\bbZ}$,
which corresponds to the ``error around the boundary".
In a two-dimensional system, for example, we have the following
when we cut the system into two by the $y$-axis.
For $0<\theta<\frac\pi 2$, 
we define a double cone $C_\theta$ by
\begin{align}\label{ctdef}
C_\theta:=
\left\{
(x,y)\in\bbZ^2\mid
|y|\le \tan \theta\cdot |x|
\right\}.
\end{align}
Furthermore, $H_L$, $H_R$, $H_U$, $H_D$  denotes half left/right and upper/lower planes,
and $C_{\theta,L}:=C_\theta\cap H_L$, $C_{\theta,R}:=C_\theta\cap H_R$.
For any $0<\theta<\frac\pi 2$,
there is $\alpha_L\in\Aut\caA_{H_L}$, $\alpha_R\in\Aut\caA_{H_R}$,
and $\Theta\in\Aut \caA_{\lmk C_\theta\rmk^c}$ such that
\begin{align}\label{tfac}
\alpha=\Ad(v)
\lmk\alpha_{L}\otimes\alpha_{R}\rmk\circ\Theta
\end{align}
Actually, $\alpha$ can be cut in  many directions simultaneously.
Factorization property is simple but strong analytical property,
which turns out to be useful in the analysis of gapped ground state phases \cite{TRI} \cite{RI} \cite{2dSPT} 
\cite{MTCo}\cite{NaOg}.

Another property we note about $\alpha\in\QAut(\caA_{\bbZ^{\nu}})$ is that it 
does not create a long-range entanglement.
For example, it satisfies the following property.
 If $A$ and $B$ are
observables localized in finite regions far away from
each other, then $\alpha$ almost preserves the tensor product form of $A\otimes B$, namely,
there are operators $\tilde A,\tilde B$ strictly localized in some finite disjoint 
areas such that
$
 \tilde A\otimes \tilde B
$
approximates $\alpha(A\otimes B)$ in the norm topology.
In fact, our $\alpha$ can be regarded as a version of a quantum circuit with finite depth,
which is regarded as a quantum circuit which does not create long-range entanglement \cite{bl}.
From this point of view, we say a state has a short-range entanglement if it is of the form
\begin{align}\label{sle}
\lmk \bigotimes_{\bm x\in\bbZ^\nu}\rho_{\bm x}\rmk\circ\alpha,
\end{align}
with infinite tensor product state $\bigotimes_{\bm x\in\bbZ^\nu}\rho_{\bm x}$ and an automorphism $\alpha\in\QAut(\caA_{\bbZ^{\nu}})$.
Otherwise, we say it has a long-range entanglement.

In physics literature, the classification of states with respect to local unitaries is considered \cite{cgw1}.
Two states are equivalent if there is a local unitary connecting them.
In our framework, these local unitaries can be understood as automorphisms in $\QAut(\caA_{\bbZ^{\nu}})$, and
 the classification in \cite{cgw1} can be reformulated as follows.
For two states $\omega_{1},\omega_{0}$ on $\caA_{\bbZ^{\nu}}$,
we write $\omega_{1}\sim_{\mathrm l.u.}\omega_{0}$
if there is an automorphism $\alpha\in\QAut(\caA_{\bbZ^{\nu}})$ such that $\omega_{1}=\omega_{0}\circ\alpha$.
This gives some equivalence relation.
From the fact that automorphisms in $\QAut(\caA_{\bbZ^{\nu}})$ do not create long-range entanglement,
this is one physically natural criterion of classification of states.

\section{The classification of gapped ground state phases}\label{clgg}
 
The automorphisms in $\QAut(\caA_{\bbZ^{\nu}})$ are of fundamental importance in the classification problem of gapped ground state phases. In a word, ground state spaces of two interactions $\Phi_0,\Phi_1\in\caP$ (Definition \ref{gapint})
are connected
to each other via such automorphisms if they are equivalent in the classification of gapped ground state phases.
In this section, we introduce such theorem, called the automorphic equivalence.
The automorphic equivalence started as Hasting's adiabatic Lemma \cite{hw} in finite-dimensional quantum mechanical system. There have been seminal mathematical polishment and
generalization after that \cite{bmns}, \cite{NSY}
in the context of the thermodynamic limit
of quantum spin systems.
Here we introduce a version in \cite{mo},
where we require the spectral gap only in the infinite systems
(i.e., the setting in section \ref{quantumspinsec}).

The classification problem of gapped ground states in infinite
systems can be roughly described as follows.

We say two interactions $\Phi_0,\Phi_1\in \caP$
are equivalent if 
there is a path
of interactions $\Phi : [0,1]\to \caP$
satisfying the following conditions
\begin{enumerate}
\item $\Phi(0)=\Phi_0$ and $\Phi(1)=\Phi_1$
\item $[0,1]\ni s\mapsto \Phi(X;s)\in\caA_X$ is continuous and piecewise $C^{1}$.
The interaction $\Phi(s)$ and its derivative are of finite range, bounded with respect to some norm uniformly in $s\in[0,1]$.
(See (ii)-(iv) of Assumption 1.2 of \cite{mo}.)
\item for each pure $\tau_{\Phi_0}$-ground state $\varphi_0$, there is a unique smooth path of 
states $\varphi_s$ where each
$\varphi_s$ is a pure $\tau_{\Phi(s)}$-ground state.
(Here, smooth means the expectation value of 
some class of elements in $\caA_{\bbZ^{\nu}}$ with respect to $\varphi_{s}$ is differentiable,
and its derivative is not too large compared to
some norm. See \cite{mo} Assumption 1.2 (vii).)
For each $s\in [0,1]$, the map $\exa\caG_{\Phi_{0}}\ni \varphi_{0}\mapsto \varphi_{s}\in \exa\caG_{\Phi_{s}}$
gives a bijection.
\item The gap is uniformly bounded from below by some $\gamma>0$ along the path, i.e.,
$\sigma(H_{\psi_s,\Phi(s)})\setminus\{0\}\subset [\gamma,\infty)$ for
all $s\in[0,1]$ and a pure
$\tau_{\Phi_s}$-ground state $\psi_s$.
\end{enumerate}
We write $\Phi_0\sim \Phi_1$ if
$\Phi_0,\Phi_1\in\caP$ are equivalent in this sense.

The automorphic equivalence in this setting is given as follows.
\begin{thm}\label{automorphic}\cite{mo}
If $\Phi_0\sim\Phi_1$, then
there is an $\alpha\in \QAut(\caA_{\bbZ^\nu})$
such that 
\begin{align}
\caG_{\Phi_1}=\caG_{\Phi_0}\circ\alpha.
\end{align}
\end{thm}
\begin{proof}
We use the notation above for $\Phi_0\sim\Phi_1$.
From Remark 1.4. of \cite{mo},
there is a path of automorphisms $\alpha_{s}\in \QAut(\caA_{\bbZ^{\nu}})$
satisfying $\varphi_{s}=\varphi_{0}\circ \alpha_{s}$
for each state $\varphi_{0},\varphi_{s}$ in (3).
This $\alpha_s$ is independent of the choice of
$\varphi_0$.
Because $\caG_{\Phi(s)}$ is a convex weak$*$-compact set, it coincides with
the weak$*$-closure of the convex hull of extremal points of $\caG_{\Phi(s)}$.
Hence we see that this $\alpha_{s}$ maps $\caG_{\Phi(0)}$ to
$\caG_{\Phi(s)}$ bijectively.
\end{proof}
Hence automorphisms in $\QAut(\caA_{\bbZ^{\nu}})$ connect ground state spaces of $\Phi_0$ and $\Phi_1$.
For this reason, this class of automorphisms
is of fundamental importance.
The point here is that it is not only that 
there is some automorphism connecting the ground
state spaces, but also, we know the details of the 
automorphisms.

Note that for interactions $\Phi_{1},\Phi_{0}\in\caP$
with unique ground states $\omega_{\Phi_{1}}$, $\omega_{\Phi_{0}}$,
 $\Phi_{1}\sim \Phi_{0}$
implies $\omega_{\Phi_{1}}\sim_{\mathrm l.u.} \omega_{\Phi_{0}}$ from Theorem \ref{automorphic}.
For the moment of writing, it is not clear for us if the converse is true.

We call an on-site interaction (defined in (\ref{onsitedef}))
with a unique gapped ground state a trivial interaction.
The unique ground state $\omega_{\Phi_0}$ of a trivial interaction $\Phi_0$
is of infinite tensor product form.
One can easily see that any two trivial interactions
are equivalent. 
The equivalence class $\caP_0$ of interactions including
these trivial interactions is called a trivial phase.
Any interaction $\Phi$ in the trivial phase
has a unique ground state, and from Theorem \ref{automorphic}, 
it has a short-range entanglement (\ref{sle}).

\section{Symmetry protected topological (SPT) phases}
The trivial phase $\caP_0$ consists of interactions
that are connected to trivial interactions,
and as a result, its ground state has a short-range entanglement and is basically the same as product states.
From this point of view, the trivial phase itself may not be that interesting.
However, if we introduce some symmetry to the game, 
we can extract some interesting mathematical structure out of it.
This is so-called symmetry protected topological (SPT) phases, which were introduced by Gu and Wen \cite{GuWen2009} \cite{cglw} \cite{ChenGuWEn2011}.
Throughout this section $\omega_\Phi$
for $\Phi\in \caP_0$ indicates the unique ground state
of $\Phi$.

In this talk, as a symmetry, we consider an on-site finite group symmetry, which is defined as follows.
(A study on the global reflection symmetry in one-dimensional systems
can be found in \cite{RI}.)
We fix a  finite group $G$ and a (projective) unitary representation $U$ of $G$
on $\bbC^d$.
Then there is a unique automorphism $\beta_{g}$ satisfying
\begin{align*}
\beta_g(A)=\lmk \bigotimes_{x\in\Lambda} U(g)\rmk A\lmk \bigotimes_{x\in\Lambda} U(g)^*\rmk,\;
g\in G,\;\;  A\in\caA_{\Lambda},\;\; \Lambda\in{\mathfrak S}_{\bbZ^\nu}.
\end{align*}
Clearly, this gives an action of $G$ on $\caA_{\bbZ^\nu}$, i.e., $\beta_{g}\beta_{h}=\beta_{gh}$
for $g,h\in G$.
We call this action of $G$, an on-site symmetry given by $G$ and $U$.
We say an interaction $\Phi$ is $\beta$-invariant
if $\beta_g(\Phi(X))=\Phi(X)$
for all $X\in {\mathfrak S}_{\bbZ^\nu}$ and $g\in G$.
For a ground state $\varphi$ of a $\beta$-invariant interaction $\Phi$, one can check that
$\varphi\circ\beta_{g}$ is also a ground state of $\Phi$.
Therefore, if a $\beta$-invariant interaction $\Phi$
has a unique ground state $\omega_{\Phi}$, the ground state is
$\beta$-invariant, $\omega_{\Phi}\circ\beta_{g}=\omega_{\Phi}$.

What we are interested in, in this section is the set of
all $\beta$-invariant interactions 
 in the trivial phase $\caP_0$.
 We denote the set of all such interactions by $\caP_{0,\beta}$.
We would like to classify them with respect to the following 
criterion.
Two interactions $\Phi_0$, $\Phi_1$ are $\beta$-equivalent
if there is a smooth path of interactions in $\caP_{0,\beta}$
satisfying the conditions (1)-(4) we saw in section \ref{clgg}.
We write $\Phi_0\sim_\beta\Phi_1$
in this case.
The difference between $\sim$ and $\sim_\beta$
is that we require the symmetry to be preserved along the path.
Because of this additional condition, 
there can be interactions $\Phi_0,\Phi_1\in\caP_{0,\beta}$,
which satisfy $\Phi_0\sim\Phi_1$ (by definition)
but not $\Phi_0\sim_\beta\Phi_1$.
In other words, $\caP_{0,\beta}$ may split into possibly
multiple equivalence classes.
The resulting equivalence classes are the
symmetry protected topological (SPT) phases.

For this SPT classification problem,
physicists and algebraic topologists have a conjecture
\cite{ktt} \cite{yonekura}.
They say that 
 SPT-phases should be understood in terms of the invertible quantum field theory. As a result,
 for a finite group $G$, SPT-phases should be classified by the Pontryagin dual of bordism group on the classifying space
  $BG$ of $G$.
 In one and two-dimensions, these Pontryagin duals are 
 {$H^2(G,\Uo)$, $H^3(G, \Uo)$}.
 In fact, we can derive these group cohomology valued invariants out of
 our general microscopic models of 
 in thoes dimensions.
\begin{thm}\label{sptthm}\cite{TRI}\cite{2dSPT}
There is a $H^{2}(G,\Uo)$-valued invariant 
for one-dimensional SPT-phases.
There is a {$H^{3}(G,\Uo)$}-valued invariant for
two-dimensional SPT-phases.
\end{thm}
For the rest of this section, we explain how to find such invariants
out of general models.
In the analysis of gapped ground state phases, there is a general guiding principle
to find an invariant. That is,
 cut the system into two and look at the edge.
 This principle is sometimes called the bulk-edge correspondence.
 In order to derive the invariant in the Theorem \ref{sptthm},
 we follow this principle and restrict our group action 
 $\beta$ to the
half of the system.
Namely, we consider the group actions
 \begin{align}
 \beta_g^R:=\id_{\caA_{L}}
\otimes \bigotimes_{x\ge 0} 
\Ad\lmk U(g)\rmk,\quad
\beta_g^U:=\id_{\caA_{H_D}}
\otimes \bigotimes_{(x,y)\in H_U}
\Ad\lmk U(g)\rmk,
 \end{align}
 in one and two dimensions, respectively. 
 We investigate the effect of these actions on our unique ground state
$\omega_\Phi$ for $\Phi\in\caP_{0,\beta}$.

Let us start with one-dimensional systems.
Recall that  $\omega_\Phi$ has a short-range entanglement, and
is $\beta$-invariant.
From these facts, we expect that
the effect of $\beta^R$ is not much recognizable on the left infinite chain, far away from the origin.
On the other hand, on the right infinite chain, far away from the origin,
the difference between $\beta$ and $\beta^R$
 are not much recognizable.
 Combining this and the fact that $\omega_\Phi$ is $\beta$-invariant,
 we conclude that the effect of $\beta^R$ is not much recognizable on the right infinite chain, far away from the origin.
 As a result, we expect that the effect of $\beta^R$
 on $\omega_\Phi$ should be localized around the origin.
 In other words, $\omega_\Phi$ and $\omega_\Phi\circ\beta^{R}_{g}$
are macroscopically the same.
It turns out to be true, mathematically, in the following sense.
\begin{prop}\label{obd}
The state $\omega_\Phi$ and $\omega_\Phi\circ\beta^{R}_{g}$
are equivalent.
\end{prop}
This can be seen very easily.
Recall from the definition that $\Phi\in\caP_{0}$
means $\Phi\sim\Phi_{0}$ with some trivial interaction $\Phi_{0}$.
From Theorem \ref{automorphic}, we have $\omega_{\Phi}=\omega_{\Phi_{0}}\circ\alpha$
with some $\alpha\in\QAut(\caA_{\bbZ})$.
Recall that as a trivial interaction, $\Phi_{0}$ has a unique ground state of infinite tensor product form.
In particular, we can write $\omega_{\Phi_{0}}$ as $\omega_{\Phi_{0}}=\omega_{L}\otimes \omega_{R}$
with pure states $\omega_{L}$, $\omega_{R}$ on the left and right infinite chains $\caA_{L}$, $\caA_{R}$,
respectively.
Recall also that our $\alpha$ satisfies the factorization property (\ref{ofac}).
Combining these, we conclude that 
\begin{align}
\omega_{\Phi}\simeq \lmk \omega_{L}\otimes \omega_{R}\rmk\circ
\lmk\alpha_{L}\otimes\alpha_{R}\rmk,
\end{align}
with some automorphisms $\alpha_{L},\alpha_{R}$ on $\caA_{L}$, $\caA_{R}$.
From this and the invariance of $\omega_{\Phi}$ under $\beta_{g}$,
we see that
$
\omega_{L}\alpha_{L}\beta_{g}^{L}\otimes \omega_{R}\alpha_{R}\beta_{g}^{R}
\simeq \omega_{L}\alpha_{L}\otimes \omega_{R}\alpha_{R},
$
where $\beta^{L}$,$\beta^R$ are the restrictions of $\beta$ to the left, right infinite chains.
This implies $\omega_{R}\alpha_{R}\beta_{g}^{R}\simeq \omega_{R}\alpha_{R}$,
 hence we get
 \begin{align}
 \omega_{\Phi}\beta_{g}^R\simeq
 \omega_{L}\alpha_{L}\otimes \omega_{R}\alpha_{R}\beta_{g}^R\simeq
 \omega_{L}\alpha_{L}\otimes \omega_{R}\alpha_{R}\simeq
 \omega_{\Phi},
 \end{align}
proving the claim.

Note from section \ref{quantumspinsec} that
Proposition \ref{obd} means $\beta^{R}_{g}$ is implementable by a unitary
$u_{g}$ in the GNS representation $(\caH_{\omega_{\Phi}},\pi_{\omega_{\Phi}})$
of $\omega_{\Phi}$, i.e.,
\begin{align}
\Ad\lmk u_{g}\rmk\circ\pi_{\omega_{\Phi}}=\pi_{\omega_{\Phi}}\circ \beta^{R}_{g}.
\end{align}
Because $\beta^{R}_{}$ is a group action, we have
\begin{align}\label{ughgh}
\Ad\lmk u_{g}u_{h}\rmk\circ\pi_{\omega_{\Phi}}=\pi_{\omega_{\Phi}}\circ \beta^{R}_{g}\beta^{R}_{h}
=\pi_{\omega_{\Phi}}\circ \beta^{R}_{gh}=\Ad\lmk u_{gh}\rmk\circ\pi_{\omega_{\Phi}},\quad g,h\in G.
\end{align}
Recall that $\omega_{\Phi}$ is a unique ground state of $\Phi$, hence it is pure.
As a result, $\pi_{\omega_{\Phi}}(\caA_{\bbZ})$
is dense in $\caB(\caH_{\omega_{\Phi}})$ with respect to the strong operator topology.
From this, (\ref{ughgh}) implies there is some $\sigma(g,h)\in \Uo$ such that
\begin{align}
u_{g}u_{h}=\sigma(g,h) u_{gh},\quad g,h\in G.
\end{align}
In other words, $(u_{g})$ forms a projective representation.
As a result, we obtain $H^{2}(G,\Uo)$-valued index out of it.

Using the automorphic equivalence Theorem \ref{automorphic} and the factorization property
of the automorphism therein,
one can show that it is in fact an invariant of our classification $\sim_{\beta}$ \cite{TRI}.
The point of the proof is, when $\Phi_{0}\sim_{\beta}\Phi_{1}$, the time-dependent interactions
giving $\alpha\in\QAut(\caA_{\bbZ})$
in Theorem \ref{automorphic} can be taken to be $\beta$-invariant.
Proposition \ref{obd} itself holds for general $\beta$-invariant unique gapped ground state.
This is thanks to the theorem by Matsui \cite{Matsui3}
showing the split property for unique gapped ground states.
Projective representations associated to split states have been known from the time around 00 \cite{Matsui2}
among operator algebraists. What is new here is that the associated 
cohomology class is an invariant of our classification.
In fact, this $H^2(G,\Uo)$-valued index is
a complete invariant of pure $\beta$-invariant split states
with respect to some classification \cite{GS}.
This index can be used to show {L}ieb-{S}chultz-{M}attis type theorems
\cite{LSM}\cite{AL}\cite{Matsui2}\cite{NS} 
(no-go theorems for the existence of
unique gapped ground state under some symmetry),
 for finite groups symmetries
\cite{ot}\cite{OTT}.

For two dimension, $\omega_{\Phi}\circ\beta^{U}_{g}$ is not equivalent to $\omega_{\Phi}$ in general.
However, an analogous argument as the one-dimensional case lets us expect
that the effect of $\beta^{U}_{g}$ should be localized around the $x$-axis.
In fact, it turns out to be true mathematically.
\begin{prop}\label{etapro}
For any $0<\theta<\frac\pi 2$, there are
$\eta_{g,L}\in \Aut\lmk \caA_{C_{\theta,L}}\rmk$ and
$\eta_{g,R}\in \Aut\lmk \caA_{C_{\theta,R}}\rmk$ such that
\begin{align*}
\omega_{\Phi}\circ \beta_g^U
\simeq
\omega_{\Phi}\lmk \eta_{g,L}\otimes\eta_{g,R}\rmk.
\end{align*}
\end{prop}
It means macroscopically, the effect of $\beta^{U}_{g}$ on $\omega_{\Phi}$
is localized around $C_{\theta,L}$ and $C_{\theta,R}$ for any $0<\theta<\frac\pi 2$.
This $\eta_{g,R}$ is our source of $H^{3}(G,\Uo)$-valued index.

Now we fix some $0<\theta<\frac\pi 2$, and set
 $\gamma_{g}^{R}:= \beta_g^{UR}\circ\eta_{g,R}^{-1}$, 
 $\gamma_{g}^{L}:= \beta_g^{UL}\circ\eta_{g,L}^{-1}$
with $\eta_{g,R}$, $\eta_{g,L}$ for this $\theta$.
Here, $\beta_g^{UR}$, $\beta_g^{UL}$
are group actions of $G$ given by
\[
\beta_g^{UR}:=\id_{\lmk H_U\cap H_{R}\rmk^{c}}
\otimes \bigotimes_{(x,y)\in H_U\cap H_{R}}
\Ad\lmk U(g)\rmk,\quad
\beta_g^{UL}:=\id_{\lmk H_U\cap H_{L}\rmk^{c}}
\otimes \bigotimes_{(x,y)\in H_U\cap H_{L}}
\Ad\lmk U(g)\rmk.
\]
From Proposition \ref{etapro},
we have 
\begin{align}\label{ginv}
\omega_{\Phi}\circ\lmk\gamma_{g}^{L}\otimes \gamma_{g}^{R}\rmk\simeq \omega_{\Phi},\quad g\in G.
\end{align}

On the other hand, recall from the definition that $\Phi\in\caP_{0}$
means $\Phi\sim\Phi_{0}$ with some trivial interaction $\Phi_{0}$.
From Theorem \ref{automorphic}, we have $\omega_{\Phi}=\omega_{\Phi_{0}}\circ\alpha$,
with $\alpha\in\QAut(\caA_{\bbZ^\nu})$
satisfying the factorization property, i.e.,
\begin{align}
\alpha=\Ad(v)\circ\lmk\alpha_{L}\otimes\alpha_{R}\rmk\circ\Theta,\quad
\alpha_{L}\in\Aut\caA_{H_{L}},\quad \alpha_{R}\in\Aut\caA_{H_{R}},\quad 
\Theta\in\Aut\caA_{C_{\theta}^{c}},
\end{align}
for our fixed $\theta$.
Recall that as a trivial interaction, $\Phi_{0}$ has a unique ground state $\omega_{\Phi_0}$ of infinite tensor product form.
In particular, we can write $\omega_{\Phi_{0}}$ as $\omega_{\Phi_{0}}=\omega_{L}\otimes \omega_{R}$
with pure states $\omega_{L}$, $\omega_{R}$ on $\caA_{H_{L}}$, $\caA_{H_{R}}$,
respectively.
Combining these, we conclude that 
\begin{align}\label{decompo}
\omega_{\Phi}\simeq \lmk \omega_{L}\otimes \omega_{R}\rmk\circ
\lmk\alpha_{L}\otimes\alpha_{R}\rmk\circ\Theta.
\end{align}

Repeated use of (\ref{ginv}) gives 
\begin{align}
\omega_{\Phi}\circ
\lmk\gamma_{g}^{L}\gamma_{h}^{L}\lmk \gamma_{gh}^{L}\rmk^{-1}\otimes 
\gamma_{g}^{R}\gamma_{h}^{R}\lmk \gamma_{gh}^{R}\rmk^{-1}\rmk\simeq \omega_{\Phi}.
\end{align}
Applying (\ref{decompo}) to this, we obtain
\begin{align}\label{ttt2}
\lmk \omega_{L}\otimes \omega_{R}\rmk\circ
\lmk\alpha_{L}\otimes\alpha_{R}\rmk\circ\Theta\circ
\lmk\gamma_{g}^{L}\gamma_{h}^{L}\lmk \gamma_{gh}^{L}\rmk^{-1}\otimes 
\gamma_{g}^{R}\gamma_{h}^{R}\lmk \gamma_{gh}^{R}\rmk^{-1}\rmk
\simeq
\lmk \omega_{L}\otimes \omega_{R}\rmk\circ
\lmk\alpha_{L}\otimes\alpha_{R}\rmk\circ\Theta.
\end{align}
 Note that 
 \begin{align}
 \gamma_{g}^{R}\gamma_{h}^{R}\lmk \gamma_{gh}^{R}\rmk^{-1}
 =\lmk \beta_g^{UR}\eta_{g,R}^{-1}\lmk \beta_g^{UR}\rmk^{-1}\rmk
 \lmk \beta_{gh}^{UR}\eta_{h,R}^{-1}\eta_{gh,R}\lmk \beta_{gh}^{UR}\rmk^{-1}\rmk
 \in \Aut\lmk\caA_{C_{\theta,R}}\rmk.
 \end{align}
 Similarly, we have $\gamma_{g}^{L}\gamma_{h}^{L}\lmk \gamma_{gh}^{L}\rmk^{-1}\in \Aut\lmk\caA_{C_{\theta,L}}\rmk$.
 Therefore, they commute with $\Theta\in \Aut\lmk\caA_{C_{\theta}^{c}}\rmk$.
 From this and (\ref{ttt2}), we obtain
 \[
 \lmk \omega_{L}\otimes \omega_{R}\rmk\circ
\lmk\alpha_{L}\otimes\alpha_{R}\rmk\circ
\lmk\gamma_{g}^{L}\gamma_{h}^{L}\lmk \gamma_{gh}^{L}\rmk^{-1}\otimes 
\gamma_{g}^{R}\gamma_{h}^{R}\lmk \gamma_{gh}^{R}\rmk^{-1}\rmk
\simeq
\lmk \omega_{L}\otimes \omega_{R}\rmk\circ
\lmk\alpha_{L}\otimes\alpha_{R}\rmk,
\]
 which implies
 \begin{align}
 \omega_{R}\alpha_{R}\gamma_{g}^{R}\gamma_{h}^{R}\lmk \gamma_{gh}^{R}\rmk^{-1}
 \simeq \omega_{R}\alpha_{R}.
 \end{align}
 Recall from section \ref{quantumspinsec},
 this means that the automorphism 
 $\gamma_{g}^{R}\gamma_{h}^{R}\lmk \gamma_{gh}^{R}\rmk^{-1}$
 is implementable by a unitary $u(g,h)$
 in the GNS representation $(\caH_{R},\pi_{R})$ of $ \omega_{R}\alpha_{R}$ i.e.,
 \begin{align}\label{uimple}
 \Ad\lmk u(g,h)\rmk\pi_{R}=\pi_{R}\gamma_{g}^{R}\gamma_{h}^{R}\lmk \gamma_{gh}^{R}\rmk^{-1}.
 \end{align}
 
 Note also that (\ref{decompo}) and (\ref{ginv}) implies 
 \begin{align}
\lmk \omega_{L}\otimes \omega_{R}\rmk\circ
\lmk\alpha_{L}\otimes\alpha_{R}\rmk\circ\Theta\circ\lmk\gamma_{g}^{L}\otimes \gamma_{g}^{R}\rmk\simeq
\lmk \omega_{L}\otimes \omega_{R}\rmk\circ
\lmk\alpha_{L}\otimes\alpha_{R}\rmk\circ\Theta.
 \end{align}
 Therefore, with $(\caH_{L},\pi_{L})$ a GNS representation of $ \omega_{L}\alpha_{L}$,
 there is a unitary $W_{g}$ on $\caH_{L}\otimes \caH_{R}$
 implementing $\Theta\circ\lmk\gamma_{g}^{L}\otimes \gamma_{g}^{R}\rmk\circ\Theta^{-1}$
 in the GNS representation $(\caH_{L}\otimes \caH_{R},\pi_{L}\otimes \pi_{R})$
 of $ \omega_{L}\alpha_{L}\otimes \omega_{R}\alpha_{R}$, i.e.,
 \begin{align}\label{wimple}
 \Ad\lmk W_{g}\rmk\lmk \pi_{L}\otimes \pi_{R}\rmk
 =\lmk \pi_{L}\otimes \pi_{R}\rmk \circ \Theta\circ\lmk\gamma_{g}^{L}\otimes \gamma_{g}^{R}\rmk\circ\Theta^{-1}.
 \end{align}
 
 For these $u(g,h)$ (\ref{uimple}) and $W_{g}$ (\ref{wimple}),
 we claim there are $c(g,h,k)\in \Uo$
 such that
 \begin{align}\label{3coc}
 \Ad\lmk W_{g}\rmk\lmk\unit_{L}\otimes u(h,k)\rmk\cdot \lmk\unit_{L}\otimes u(g,hk)\rmk
 =c(g,h,k)\lmk\unit_{L}\otimes u(g,h)u(gh,k)\rmk,\quad g,h,k\in G.
 \end{align}
 To see this, consider $\pi_{L}\otimes\pi_{R}\gamma_{g}^{R}\gamma_{h}^{R}\gamma_{k}^{R}$.
 On the one hand,  
 with the repeated use of (\ref{uimple}), we have
 \begin{align}\label{oohd}
 \pi_{L}\otimes\pi_{R}\gamma_{g}^{R}\gamma_{h}^{R}\gamma_{k}^{R}
 =\Ad\lmk\unit_{L}\otimes u(g,h)\rmk
 \lmk \pi_{L}\otimes\pi_{R}\gamma_{gh}^{R}\gamma_{k}^{R}\rmk
  =\Ad\lmk\unit_{L}\otimes u(g,h)u(gh,k)\rmk
 \lmk \pi_{L}\otimes\pi_{R}\circ\gamma_{ghk}^{R}\rmk.
 \end{align}
 On the other hand, 
 note that  both of $ \gamma_{h}^{R}\gamma_{k}^{R}\lmk \gamma_{hk}^{R}\rmk^{-1}$ and 
$\gamma_{g}^{R}\lmk \gamma_{h}^{R}\gamma_{k}^{R}\lmk \gamma_{hk}^{R}\rmk^{-1}\rmk\lmk \gamma_{g}^{R}\rmk^{-1}$
commute with $\Theta$ as before.
Hence we have 
\begin{align}
\id_{L}\otimes
\gamma_{g}^{R}\lmk \gamma_{h}^{R}\gamma_{k}^{R}\lmk \gamma_{hk}^{R}\rmk^{-1}\rmk\lmk \gamma_{g}^{R}\rmk^{-1}
=\Theta\lmk
\gamma_{g}^{L}\otimes \gamma_{g}^{R}\rmk\Theta^{-1}
\lmk \id_{L}\otimes 
\gamma_{h}^{R}\gamma_{k}^{R}\lmk \gamma_{hk}^{R}\rmk^{-1}\rmk
\Theta\lmk \gamma_{g}^{L}\otimes {\gamma_{g}^{R}}\rmk^{-1}\Theta^{-1}.
\end{align}
From this and repeated use of (\ref{uimple}), (\ref{wimple}), we have
 \begin{align}
 \begin{split}
 &\pi_{L}\otimes\pi_{R}\gamma_{g}^{R}\gamma_{h}^{R}\gamma_{k}^{R}\\
& =\lmk \pi_{L}\otimes\pi_{R} \rmk
 \Theta\lmk
\gamma_{g}^{L}\otimes \gamma_{g}^{R}\rmk\Theta^{-1}
\lmk 
\id_{L}\otimes\gamma_{h}^{R}\gamma_{k}^{R}\lmk \gamma_{hk}^{R}\rmk^{-1}\rmk
\Theta\lmk \gamma_{g}^{L}\otimes {\gamma_{g}^{R}}\rmk^{-1}\Theta^{-1}
\lmk \id_{L}\otimes \gamma_{g}^{R}\gamma_{hk}^{R}\rmk
 \\
&=\Ad\lmk W_{g} \lmk \unit_{L}\otimes u(h,k)\rmk W_{g}^{*}\lmk
\unit_{L}\otimes u(gh,k)\rmk\rmk\lmk \pi_{L}\otimes\pi_{R}\gamma_{ghk}^{R} \rmk.
\end{split}
 \end{align}
 Comparing this and (\ref{oohd}), we have
 \begin{align}\label{cocoa}
 \Ad\lmk\unit_{L}\otimes u(g,h)u(gh,k)\rmk
 \lmk \pi_{L}\otimes\pi_{R}\rmk
 =\Ad\lmk W_{g} \lmk \unit_{L}\otimes u(h,k)\rmk W_{g}^{*}\lmk
\unit_{L}\otimes u(gh,k)\rmk\rmk\lmk \pi_{L}\otimes\pi_{R} \rmk.
 \end{align}
 Note, because $(\caH_{L}\otimes \caH_{R},\pi_{L}\otimes \pi_{R})$
 is a GNS representation of a pure state $ \omega_{L}\alpha_{L}\otimes \omega_{R}\alpha_{R}$,
 $\lmk \pi_{L}\otimes \pi_{R}\rmk\lmk\caA_{\bbZ^{2}}\rmk$ is dense in
 $\caB\lmk \caH_{L}\otimes \caH_{R}\rmk$ with the strong operator topology.
 As a result, (\ref{cocoa}) implies our claim (\ref{3coc}).
 
 The situation in (\ref{uimple}),  (\ref{3coc}) is pretty much similar to that of cocycle actions
  \cite{Connes},\cite{jones}.
 In fact, following the argument in \cite{jones},
 we can show that $c(g,h,k)$ satisfies the $3$-cocycle relation.
 Hence, out of it, we obtain a $H^{3}(G,\Uo)$-valued index.
 Using the automorphic equivalence Theorem \ref{automorphic} and the factorization property
of the automorphism therein,
one can show that it is in fact an invariant of our classification $\sim_{\beta}$.

A derivation of indices for {SPT}-phases was
initially carried out in {tensor network} models, matrix product states
{MPS}, \cite{po}, \cite{po2}, \cite{Perez-Garcia2008} in one dimension, and 
projected entangled pair states \cite{molnar}.
Our indices coincide with theirs in those models.
In other words, thanks to those works, there are many examples.
Our approach introduced in this section is operator
algebraic.
Recently, some quantum information based approach are
reported \cite{ksy}\cite{so}.
 
\section{Anyons in topological phases}
In this section, we consider the classification $\sim_{\mathrm l.u.}$ in two-dimension.
Recall that states which are equivalent to an infinite tensor product state 
with respect to $\sim_{\mathrm l.u.}$ are said to have a short-range entanglement, and otherwise
it is said to have a long-range entanglement.
It is frequently said that in the two dimensional systems,
the existence of ``anyon'' means the long-range entanglement of the state.
In this section, we formulate this statement in our operator algebraic setting.

An anyon is a string-like excitation with a braiding structure.
How to formulate an anyon mathematically is a non-trivial question of mathematical physics.
Our answer, motivated by AQFT and studies of Kitaev models \cite{N1}\cite{N2}\cite{FN}\cite{CNN1}
 is that it is a superselection sector.
 It is defined in terms of cones.
By a cone, we mean a subset of $\bbZ^{2}$ of the form
\begin{align*}
\Lambda_{\bm a, \theta,\varphi}
:=&\left\{
\bm x\in\bbZ^{2}\mid (\bm x-\bm a)\cdot \bm e_{\theta}>\cos\varphi\cdot \lV \bm x-{\bm a}\rV
\right\},
\end{align*}
with some $\bm a\in \bbR$, $\theta\in\bbR$ and $\varphi\in (0,\pi)$.
Here we set $\bm e_\theta:=(\cos\theta,\sin\theta)$.
For a cone $\Lambda:=\Lambda_{\bm a, \theta,\varphi}$ 
and $\bm b\in \bbR^2$, $\varepsilon>0$, we set 
$\Lambda_{\varepsilon}+\bm b:=\Lambda_{\bm a+\bm b, \theta,\varphi+\varepsilon}$,
$|\arg\Lambda|:=2\varphi$
and $\bm e_{\Lambda}:=\bm e_\theta$.
\begin{defn}
Let $(\caH,\pi_0)$ be an irreducible representation of $\caA_{\bbZ^2}$.
We say a representation $\pi$ of $\caA_{\bbZ^2}$ on $\caH$
satisfies the superselection criterion for $\pi_0$ if 
\[
\pi\vert_{\caA_{\Lambda^c}}\simeq_{u.e.}
 \pi_0\vert_{{\caA_{\Lambda^c}}},
\]
for any cone $\Lambda$ in $\bbZ^2$.
(Here, $\simeq_{u.e.}$ means that the two representations are unitarily equivalent.)
Such representations are called superselection sectors for $\pi_0$.
\end{defn}
Super selection sectors are objects studied extensively in AQFT.
In the context of quantum spin systems, P. Naaijkens and his coauthors
carried out studies on 
Kitaev's quantum double model 
from the point of view of superselection sectors \cite{N1}\cite{N2}\cite{FN}\cite{CNN1},
where they drove braiding structure.

We can see the importance of the sector theory for us from the fact that it is an invariant of $\sim_{\mathrm l.u.}$.
 \begin{thm}\label{thm:invariant}\cite{NaOg}
Let $(\caH,\pi_0)$ be an irreducible representation
and let $\alpha\in \QAut(\caA_{\bbZ^{2}})$.
Suppose that a representation $\pi$ satisfies the superselection criterion
for $\pi_0$. Then  $\pi \circ \alpha$  satisfies the superselection criterion
for $\pi_0 \circ \alpha$
 \end{thm}
 Let $\omega_1$, $\omega_0$ be pure states such that 
 $\omega_{1}\sim_{\mathrm l.u.}\omega_{0}$ 
 with $\omega_{1}=\omega_{0}\circ\alpha$, $\alpha\in \QAut\lmk\caA_{\bbZ^{2}}\rmk$.
 Then from Theorem \ref{thm:invariant}, 
  $\alpha$ gives a bijection between the set of all superselection sectors of $\pi_{\omega_{0}}$
 and the set of all superselection sectors of $\pi_{\omega_{1}}$.
 
The proof of \ref{thm:invariant} is a simple argument using the factorization property.
For $\varepsilon>0$,
analogous to (\ref{tfac}),
we have a decomposition 
\begin{align}
\alpha=\Ad\lmk v\rmk\circ\Xi\circ\lmk \alpha_{\Lambda}\otimes \alpha_{\Lambda^{c}}\rmk,
\end{align}
where $\alpha_{\Lambda}$, $\alpha_{\Lambda^{c}}$,
$\Xi$ are automorphisms on $\caA_{\Lambda}$, $\caA_{\Lambda^{c}}$,$\caA_{\Lambda_{\varepsilon}}$,
respectively. (We choose $\varepsilon>0$ small
enough so that $\Lambda_\varepsilon$ is still a cone.)
Then for a superselection sector $\pi$ for $\pi_{0}$, we have
\begin{align}
\pi\circ\alpha\vert_{\caA_{\Lambda}}\sim_{\mathrm u.e.}
\pi\circ \Xi\circ \alpha_{\Lambda}\vert_{\caA_{\Lambda}}
=\pi\vert_{\caA_{\Lambda_{\varepsilon}}}\circ \Xi\circ\alpha_{\Lambda}\vert_{\caA_{\Lambda}}
\sim_{\mathrm u.e.}
\pi_{0}\vert_{\caA_{\Lambda_{\varepsilon}}}\circ \Xi\circ \alpha_{\Lambda}\vert_{\caA_{\Lambda}}
\sim_{\mathrm u.e.}\pi_{0}\circ\alpha\vert_{\caA_{\Lambda}},
\end{align}
proving the claim.

We say that $\pi_0$ has a trivial sector theory if 
any representation satisfying the superselection criterion for $\pi_0$
is quasi-equivalent to $\pi_0$.
Otherwise, we say $\pi_0$ has a non-trivial sector theory.
One can show that for a pure state of infinite tensor product form, its GNS representation
 has a trivial sector theory \cite{NaOg}.
Combing this and Theorem \ref{thm:invariant},
we obtain the following.
\begin{cor}
If a pure state has a short-range entanglement, then
its GNS representation has a trivial sector theory.
\end{cor}

In other words, the existence of non-trivial superselection sectors implies the long-range entanglement.
If we regard superselection sectors as anyons, it is a mathematical
realization of the folklore saying that the existence of
anyons implies long-range entanglement of the state.

The reason why we expect superselection sectors
to be related to anyons comes from AQFT.
Using the tools from AQFT, in \cite{CNN2} Cha-Naaijkens-Nachtergaele
 derived a braiding structure in a general setting of
 semi-group of almost localized endomorphisms in quantum spin systems.
It is well known that anyons show up in AQFT 
surprisingly naturally \cite{BDMRS} \cite{BF} \cite{DHRI}\cite{FRS}\cite{K}\cite{BKLR}.
More precisely, under some condition called Haag duality, 
a braided $C^{*}$-tensor category can be associated to the irreducible representation
with non-trivial sector theory.
The Haag duality is the property
$\pi_0(\caA_{\Lambda^c})'=\pi_0(\caA_{\Lambda})''$,
for all cones $\Lambda$ in $\bbZ^{2}$.

The problem for us about introducing this condition in quantum spin systems
is that it does not look to be plausible that this condition is stable under 
automorphisms in $\QAut(\caA_{\bbZ^2})$.
Recalling that automorphisms in $\QAut(\caA_{\bbZ^2})$
are 
the fundamental operation in the classification problem of gapped ground state phases,
this situation is not convenient for us.
For this reason, we introduce a weaker version of Haag duality.
\begin{defn}\label{assum7}[Approximate Haag duality \cite{MTCo}]
Let $(\caH,\pi_0)$ be an irreducible representation of $\caA_{\bbZ^{2}}$.
We say that $(\caH,\pi_0)$ satisfies the approximate Haag duality if
the following conditions hold.:
For any $\varphi\in (0,2\pi)$ and 
 $\varepsilon>0$ with
$\varphi+4\varepsilon<2\pi$,
there is some $R_{\varphi,\varepsilon}>0$ and decreasing
functions $f_{\varphi,\varepsilon,\delta}(t)$, $\delta>0$
on $\bbR_{\ge 0}$
with $\lim_{t\to\infty}f_{\varphi,\varepsilon,\delta}(t)=0$
such that
\begin{description}
\item[(i)]
for any cone $\Lambda$ with $|\arg\Lambda|=\varphi$, there is a unitary 
$U_{\Lambda,\varepsilon}\in \caU(\caH)$
satisfying
\begin{align}\label{lem7p}
\pi_0\lmk\caA_{\Lambda^c}\rmk'\subset 
\Ad\lmk U_{\Lambda,\varepsilon}\rmk\lmk 
\pi_0\lmk \caA_{\lmk \Lambda-R_{\varphi,\varepsilon}\bm e_\Lambda\rmk_\varepsilon}\rmk''
\rmk,
\end{align}
and 
\item[(ii)]
 for any $\delta>0$ and $t\ge 0$, there is a unitary 
 $\tilde U_{\Lambda,\varepsilon,\delta,t}\in \pi_0\lmk \caA_{\Lambda_{\varepsilon+\delta}-t\bm e_{\Lambda}}\rmk''$
 satisfying
\begin{align}\label{uappro}
\lV
U_{\Lambda,\varepsilon}-\tilde U_{\Lambda,\varepsilon,\delta,t}
\rV\le f_{\varphi,\varepsilon,\delta}(t).
\end{align}
\end{description}
\end{defn}
The good point about this weaker version is that we know it is stable under automorphisms 
in $\QAut(\caA_{\bbZ^2})$.
\begin{prop}\label{staah}
Let $(\caH,\pi_0)$ be an irreducible representation of $\caA_{\bbZ^{2}}$
satisfying the approximate Haag duality.
Then for any automorphism $\alpha\in \QAut(\caA_{\bbZ^2})$,
 $(\caH,\pi_0\circ\alpha)$ also satisfies the approximate Haag duality.
\end{prop}
It turns out that even with this weaker version of Haag duality and 
the setting of gapped ground state phases
(which is different from that of AQFT), we can still derive a braided $C^{*}$-tensor category
(see \cite{NT} for the definition)
out of superselection sectors where, unlike endomorphisms, the multiplication rule is not apriori given
\cite{MTCo}.
The proof is a modification of the argument in AQFT and some additional argument using the gap condition
Definition \ref{gapint}.
More precisely, let $\Phi$ be a uniformly bounded finite range interaction on $\caA_{\bbZ^{2}}$ with
gapped ground states.
Let $\omega$ be a pure $\tau_\Phi$-ground state with a GNS representation
$(\caH,\pi_0,\Omega)$.
We assume that $\pi_0$ has a non-trivial sector theory, and
$\pi_{0}$ satisfies the approximate Haag duality.
Fix some $\theta\in\bbR$ and $\varphi\in (0,\pi)$ and denote by
$\ctv$ the set of all cones whose angle does not intersects with $[\theta-\varphi,\theta+\varphi]$.
We set
\begin{align}\label{btvdef}
\caB_{(\theta,\varphi)}:=
\overline{\cup_{\Lambda\in\caC_{(\theta,\varphi)} } \pi_0\lmk\caA_{\Lambda^c}\rmk' }.
\end{align}
Here $\overline{\cdot}$ denotes the norm closure.
Using the approximate Haag duality, using the argument in \cite{BF}, each super selection sectors
$\rho : \caA_{\bbZ^{2}}\to \caB\lmk\caH_{\omega}\rmk$ for
$\pi_0$ extends to an
endomorphism on $\caB_{(\theta,\varphi)}$.
We denote the extension by the same symbol $\rho$.
Via these extensions, we can introduce compositions between superselection sectors.
With this composition as a tensor,
 the superselection sectors of $\pi_0$ are the objects of our braided $C^*$-tensor category.
Our morphisms are given by the intertwiners.
Namely, for objects $\rho,\sigma$,
the morphisms from $\rho$ to $\sigma$ are bounded operators $R$ on $\caH$
such that 
$R\rho(A)=\sigma(A)R$, for any $A\in\caA_{\bbZ^2}$.
The set of all morphisms from $\rho$ to $\sigma$ is denoted by $(\rho,\sigma)$.
Note that $(\rho,\sigma)$ is a Banach space and $(\rho,\rho)$ is a $C^{*}$-algebra.
Following AQFT, the tensor of morphisms $R_{1}\in (\rho_{1},\sigma_{1})$, $R_{2}\in (\rho_{2},\sigma_{2})$
are defined by 
\begin{align}\label{mordef}
R_1\otimes R_2:=R_1{\rho}_{1}(R_2)\in \lmk\rho_{1}\otimes \rho_{2}, \sigma_{1}\otimes \sigma_{2}\rmk.
\end{align}
In fact, each intertwiner belongs to $\btv$ that ${\rho}_{1}(R_2)$
is well-defined.
Using the gap inequality and the non-triviality of the sector theory, we can show
for any cone $\Lambda$ that $\pi_{0}\lmk\caA_{\Lambda}\rmk''$ is either type $II_\infty$
or type $III$ factor. It means that 
there are isometries $u_{\Lambda}, v_{\Lambda}\in \pi_{0}\lmk\caA_{\Lambda}\rmk''$
such that $u_{\Lambda} u_{\Lambda}^{*}+v_{\Lambda} v_{\Lambda}^{*}=\unit$.
Using this, for any superselection sectors $\rho,\sigma$, we can define their direct sum
$\rho\bigoplus\sigma : {\caA}_{\bbZ^{2}}\to \caB(\caH_{0})$ by 
\begin{align}
\lmk \rho\bigoplus\sigma\rmk(A):=u_{\Lambda}\rho(A) u_{\Lambda}^{*}
+v_{\Lambda}\sigma(A) v_{\Lambda}^{*},\quad A\in{\caA}_{\bbZ^{2}} .
\end{align}
From the same fact, we can also define
subobjects.
Namely, if $p\in (\rho,\rho)$ is a non-zero projection, 
we can find some super selection sector $\sigma$
and an isometry $v$ such that $vv^{*}=p$ and $\rho(A)v=v\sigma(A)$
for all $A\in \caA_{\bbZ^{2}}$.
Hence we obtain the following theorem.
\begin{thm}\cite{MTCo}
In the above setting, superselection sectors of $\pi_{0}$
 form a braided $C^*$-tensor category.
 If two of such states $\omega_{\Phi_1},\omega_{\Phi_2}$
 satisfy $\omega_{\Phi_1}\sim_{\mathrm l.u.} \omega_{\Phi_2}$,
 then  corresponding braided $C^*$-tensor categories
 are monoidally equivalent.
\end{thm}




\end{document}